\documentclass{article}
\newcommand{\setmargins}[1]{ 
	\setlength{\textwidth}{8.5in} 
	\addtolength{\textwidth}{-#1}  \addtolength{\textwidth}{-#1}
	\setlength{\oddsidemargin}{-1in} \setlength{\evensidemargin}{-1in} 
	\addtolength{\oddsidemargin}{#1} \addtolength{\evensidemargin}{#1}
	}
\setmargins{1in}
\newcommand{\twosp}{} \newcommand{\onesp}{} 

\newenvironment{singlespace}{\onesp}{\twosp}
\newtheorem{theorem}{Theorem}

\newtheorem{lemma}[theorem]{Lemma}     
\newcounter{qnum}

\DeclareOldFontCommand{\rm}{\normalfont\rmfamily}{\mathrm}
\DeclareOldFontCommand{\sf}{\normalfont\sffamily}{\mathsf}
\DeclareOldFontCommand{\tt}{\normalfont\ttfamily}{\mathtt}
\DeclareOldFontCommand{\bf}{\normalfont\bfseries}{\mathbf}
\DeclareOldFontCommand{\it}{\normalfont\itshape}{\mathit}
\DeclareOldFontCommand{\sl}{\normalfont\slshape}{\@nomath\sl}
\DeclareOldFontCommand{\sc}{\normalfont\scshape}{\@nomath\sc}

\def\id#1{\ensuremath{\mathit{#1}}}

\def\idbf#1{\ensuremath{\mathbf{#1}}}

\def\cramped 
  {\parskip\@outerparskip\@topsep\parskip 
  \@topsepadd2pt\itemsep0pt}

\newenvironment{proof}{\trivlist\item[]\emph{Proof}:}%
{\unskip\nobreak\hskip 1em plus 1fil\nobreak$\Box$
\parfillskip=0pt%
\endtrivlist}
{\unskip\nobreak\hskip 1em plus 1fil\nobreak$\Box$
\parfillskip=0pt%
\endtrivlist}
{\unskip\nobreak\hskip 2em plus 1fil\nobreak$\Box$
\parfillskip=0pt%
\endtrivlist}
{\unskip\nobreak\hskip 2em plus 1fil\nobreak$\Box$
\parfillskip=0pt%
\endtrivlist}

\def\boxit#1{\vbox{\hrule\hbox{\vrule\kern3pt
  \vbox{\kern3pt#1\kern3pt}\kern3pt\vrule}\hrule}}
\def\Box{\boxit{\null}}

\newcommand{\zero}{\idbf{0}}
\newcommand{\one}{\idbf{1}}
\newcommand{\A}{\texttt{A}}
\newcommand{\B}{\texttt{B}}
\newcommand{\select}{\id{select}}
\newcommand{\search}{\id{search}}
\newcommand{\rank}{\id{rank}}
\newcommand{\rankzero}{\id{rank}_\zero}
\newcommand{\rankone}{\id{rank}_\one}
\newcommand{\selectzero}{\id{select}_\zero}

\newcommand{\rankb}{\id{rank}_b}
\newcommand{\selectb}{\id{select}_b}
\newcommand{\flip}{\id{flip}}

\newcommand{\ins}{\id{insert}}
\newcommand{\delete}{\id{delete}}
\newcommand{\del}{\id{delete}}

\newcommand{\qselect}{\id{qselect}}
\newcommand{\qsearch}{\id{qsearch}}

\DeclareOldFontCommand{\rm}{\normalfont\rmfamily}{\mathrm}
\DeclareOldFontCommand{\sf}{\normalfont\sffamily}{\mathsf}
\DeclareOldFontCommand{\tt}{\normalfont\ttfamily}{\mathtt}
\DeclareOldFontCommand{\bf}{\normalfont\bfseries}{\mathbf}
\DeclareOldFontCommand{\it}{\normalfont\itshape}{\mathit}
\DeclareOldFontCommand{\sl}{\normalfont\slshape}{\@nomath\sl}
\DeclareOldFontCommand{\sc}{\normalfont\scshape}{\@nomath\sc}

\def\id#1{\ensuremath{\mathit{#1}}}

\def\idbf#1{\ensuremath{\mathbf{#1}}}

\def\cramped 
  {\parskip\@outerparskip\@topsep\parskip 
  \@topsepadd2pt\itemsep0pt}

\def\boxit#1{\vbox{\hrule\hbox{\vrule\kern3pt
  \vbox{\kern3pt#1\kern3pt}\kern3pt\vrule}\hrule}}
\def\Box{\boxit{\null}}

\newcommand{\pivot}{\texttt{pivot}}
\begin{document}

\title{Online Sorting via Searching and Selection}


\author{Ankur Gupta \\
  Computer Science and Software Engineering \\
  Butler University, Indianapolis Indiana, USA, 46208 \\
  \texttt{agupta@butler.edu}
\and
Anna Kispert \\
  Computer Science and Software Engineering \\
  Butler University, Indianapolis Indiana, USA, 46208
\and
Jonathan P. Sorenson\\
  Computer Science and Software Engineering \\
  Butler University, Indianapolis Indiana, USA, 46208 \\
  \texttt{sorenson@butler.edu} \\
  \texttt{http://www.butler.edu/$\sim$sorenson}
}

\date{\today}

\maketitle

\begin{abstract}
In this paper, we present a framework
  based on a simple data structure and parameterized algorithms
  for the problems of 
    finding items in an unsorted list of linearly ordered items 
    based on their rank (selection) or value (search). 
As a side-effect of answering these online selection and search queries, 
  we progressively sort the list. 
Our algorithms are based on Hoare's Quickselect,
  and are parameterized based on the pivot selection method.

For example, if we choose the pivot as the last item in a subinterval,
  our framework yields algorithms that
  will answer $q\le n$ unique selection and/or search
  queries in a total of $O(n \log q)$ average time.
After $q=\Omega(n)$ queries the list is sorted.
Each repeated selection query takes constant time, and
  each repeated search query takes $O(\log n)$ time.
The two query types can be interleaved freely.

By plugging different pivot selection methods into our framework, 
  these results can, for example, 
  become randomized expected time or deterministic worst-case time.

We extend the algorithms and data structures in our framework
  to obtain results that are cache-oblivious I/O efficient and/or dynamic
  and/or compressed.


Our methods are easy to implement, and we show they perform well in practice.

\end{abstract}


\section{Introduction\label{sec:intro}}

Sorting and searching are two of the most fundamental problems in 
  Computer Science; they have been studied extensively 
  for decades~\cite{CLRS,Knuthv3}. 
In this paper, we present a framework that includes algorithms based on
  Hoare's well-known quickselect algorithm~\cite{Hoare61,FR75,KMP97} 
  together with a simple bitvector data structure
  for performing online selection and search queries 
  on an initially unsorted list of linearly ordered items.
The list is progressively sorted as queries are answered, with the result
  that subsequent queries become less expensive to process.
The parameter in our framework is the method used to select the
  pivot element for partitioning, and varying this method leads to
  different running time bounds.

In the literature, there are many algorithms for selection and
  multiple selection (where selection queries are known \emph{a priori}),
  based on Hoare's quickselect, 
  that largely vary only in how pivots are chosen.
See, for example,
  \cite{Alsuwaiyel2006,KMMS05,Kuba06,LM96,Panholzer03,Prodinger95,PP99}.
(Note that \cite{KMMS05} gives a thorough background on selection and
  multiple selection algorithms.)
We show that with our framework, all of these results that are based on 
  quickselect can be performed in an online fashion,
  and can include search queries as well,
  with no significant penalty in the running time analysis.
We also show that the algorithms and data structures of
  our framework can be extended to be cache-oblivious I/O efficient, 
  dynamic (allowing a modest number of updates), and compressed.

To be more specific, given a list of $n$ linearly ordered, but unsorted items
  as input in the form of an array,
  we show how to use a simple bit vector data structure to keep
  track of all previously chosen pivots that now occupy their correct
  position in the array.
By first checking this bit vector, a selection query can start its work
  in the unsorted subinterval of the array where its answer lies.
Our searching algorithm works in the same way -- it first performs a binary
  search to find the unsorted subinterval holding the answer, and then uses
  a slightly modified quickselect to search from there.
It should then be obvious that the search and select query algorithms 
  are furthering the work of a quicksort algorithm in stages, 
  performing exactly the same work, though perhaps not in the same order,
  and hence having the same overall complexity (ignoring the time
  spent manipulating the bit vector).

Let $q$ be the number of distinct selection and/or search queries 
  (repeated queries take relatively insignificant time).
By varying our method for choosing pivots, we can make use of the
  analyses performed by previous researchers to obtain results such as these:
\begin{itemize}
  \item If we choose the last item of a subinterval as the pivot
    (as is done in a typical undergraduate data structures textbook) 
    we show that $q$ queries take $O(n \log q)$ deterministic average time.%
\footnote{In this paper, we use the notation 
      $\log_b^c a = (\log_b a)^c = (\log a/ \log b)^c$ 
    to denote the $c$th power of the base-$b$ logarithm of $a$. 
    If no base is specified, the implied base is 2.}
  \item If we use a median-of-medians approach to choosing pivots
    (such as is done in \cite{BFPRT73}) we obtain an
    $O( n \log q)$ deterministic worst-cast time.
  \item If we use the randomized pivot selection algorithm as described in
    \cite{KMMS05} we obtain an $O(n\log q)$ expected running time, 
    and furthermore, we also match their bound of $B+O(n)$ 
    expected comparisons, where $B$ is the
    information-theoretic minimum number of comparisons as they define it,
    if we allow only selection queries.
    \begin{enumerate}
    \item
    Note: To include search queries in this result, 
    it suffices that we have at most $O(n/\log n)$ searches in order 
    to absorb the cost of binary search in the $O(n)$ error term, and
    the rank of the item searched for is used in computing $B$.
    However, we are then treating search queries as selection queries,
    with added cost, for defining the information-theoretic minimum, and this
    surely is not correct.  
    Addressing this issue is beyond the scope of this paper.
    \item
    Note: \cite{KMMS05} also presents a deterministic multi-selection
    algorithm, not directly based on quickselect, that uses
    $B+o(B)+O(n)$ comparisons were $B$ is as before.
    Our online framework, as presented here, does not directly apply to this
    algorithm,
    but with some significant modifications it does, and we
    will describe how that is done in the full version of this paper.
    \end{enumerate}

\end{itemize}
We present our selection algorithm in Section \ref{sec:select}, including
  a brief proof of the average case analysis mentioned above.
Search queries are described in Section \ref{sec:search}.
Our framework is easy to implement, and we present some preliminary
  experimental results in Section \ref{sec:implement}.

Sibeyn~\cite{sibeyn:external-selection} briefly considered multi-selection in 
  external memory~\cite{vitter:io-survey}, 
  but was not able to achieve optimal bounds in either the 
  \emph{a priori} case or the online case. 
Franceschini and Grossi~\cite{grossi:no-sorting} recently improved upon a 
  long line of work (listed in their paper) and proved that unsorted ordering 
  for multidimensional keys can yield optimal search time. 
\begin{itemize}
\item
We are able to extend our framework such that our algorithms are cache-oblivious
  I/O efficient.  This is discussed in Section \ref{sec:io}.
\item
We allow updates: for $\epsilon>0$, each update operation costs
  $O( (1/\epsilon) n^\epsilon)$ time, so that we can allow about
  $n^{1-\epsilon}$ online insertion and/or deletion operations 
  without significantly affecting our overall running time.
  This is described in Section \ref{sec:dynamic}.
\item
We are also able to compress our data structures, again in an online fashion,
  such that as queries are processed, the space used by our data structures
  shrink.
  This is described in Section \ref{sec:compress}.
\end{itemize}
\nocite{PB,GK}

\section{The Framework\label{sec:framework}}

In this section, we describe the algorithms in our framework.
The input is an unsorted list or array \A\ of linearly ordered items. 
We assume $\A$ contains no duplicates. 
As successive queries are answered, the items in \A\ are rearranged such that, 
  over time, \A\ becomes sorted. 
Let $n$ denote the length of \A, with the $i$th item at $\A[i-1]$.

\subsection{Pivot Function.}
We assume a pivot selection function, \pivot(),
  for use with quickselect, has been provided. 
\pivot() takes as inputs \A\ and integers $x,y$ with $0\le x\le y<n$, 
  and returns with the pivot, 
  chosen from $\A[x],\ldots,\A[y]$, swapped into $\A[y]$.
This function may rearrange the data in this portion of \A, but
  no other changes are made.
The empty function is a valid pivot function.

\subsection{Preprocessing Phase.}
Create a bit vector \B\ also of length $n$, initialized to zeroes.
Then find the minimum and maximum items in \A\ and swap them to
  positions $\A[0]$ and $\A[n-1]$ respectively.
Then set $\B[0] = \B[n-1] = \one$. 
This takes linear time.

\subsection{Bit Vector Invariant.}
\label{subsec:invariant}
If $\B[i]=1$, then for $j<i$, we have $\A[j]<\A[i]$, 
  and for $j>i$, we have $\A[j]>\A[i]$.
Obviously, this invariant holds after our preprocessing phase. 


\subsection{Selection Query Algorithm (\qselect)\label{sec:select}}

Our input is an integer~$k$, with $0\le k < n$, 
  and the goal is to find an item $a$ from \A\ such that if \A\ 
  were sorted in non-decreasing order, $a=\A[k]$.
We also rearrange the items in \A\ so that, in fact, $a=\A[k]$,
  and also $\B[k]=1$ while maintaining the invariant.

\begin{enumerate}
  \item If $\B[k]=1$, return $a=\A[k]$ and terminate.
  \item Scan \B\ to the left and right of position $k$ to find
    bit positions $\ell$ and $r$ such that $\ell<k<r$ with
    $\B[\ell]=\B[r]=1$ and, for every $i$, $\ell<i<r$ implies $\B[i]=0$.
    The preprocessing phase guarantees $\ell$ and $r$ exist.
  \item Perform Hoare's recursive quickselect algorithm on the list
    $\A[\ell+1], \ldots, \A[r-1]$.
    (So \pivot() is called using $x=\ell+1$ and $y=r-1$.)
    Whenever a pivot element is placed in its correct position in \A\ by
    the partition function, 
    set the pivot's corresponding bit to $\one$ in \B.
    This preserves the invariant.
  \item Hoare's algorithm will find and return $a$ as described above.
\end{enumerate}
Note that queries with a value of $k$ that is repeated or, by happenstance, 
  is a position of a previously chosen pivot, 
  are processed in constant time with no data comparisons.

\subsection{Average-Case Analysis}

Here we apply our framework with a pivot selection function that
  merely returns $\A[y]$ as the pivot, and determine the running time,
  averaged over permutations of the input, of processing $q$ distinct
  selection queries.

The key to our analysis is to observe that, as queries are answered,
  the algorithm is doing exactly the same work, 
  in the sense that it chooses the same pivots from subintervals of
  the same position and size,
  as the quicksort algorithm \cite{Hoare62}, 
  although not necessarily in the same order.

Let $H_k$ denote the $k$th harmonic number, and recall that $H_k=O(\log k)$.
Let $T(n,q)$ denote the average number of item comparisons performed
  to answer $q$ distinct selection queries on a list \A\ of length $n$,
  after preprocessing, using \qselect\ with the trivial pivot function.

\begin{theorem}
  $T(n,q) \le 2nH_q = O(n \log q)$.
\end{theorem}

\begin{proof}
On a random input permutation of $n$ items, the $q$ queries can be split
over a pivot in $q+1$ ways ($k$ to the left, $q-k$ to the right, $0\le k\le q$).

We ignore the case when one of the $q$ queries matches a pivot.
Such a query gets processed in constant time,
  and as $q\le n$, this is at most a linear cost total, using no
  item comparisons, so we can simply rule them out and use a smaller
  value for $q$ that only counts queries that are not previous pivots.
Using this simplification, we get the recurrence
$$
T(n,q) = n +\frac{1}{n(q+1)} 
  \sum_{j=0}^{n-1} \sum_{k=0}^q (T(j,k) + T(n-j-1, q-k)).
$$
A simple change of variables give us
$$
T(n,q) = n +\frac{2}{n(q+1)} \sum_{j=0}^{n-1} \sum_{k=0}^q  T(j,k).
$$
Observe that if we set $q=1$ and recall that $T(n,0)=0$, 
  we obtain the recurrence for Hoare's quickselect algorithm.

Dropping the $T(j,0)$ terms, which are zero, this leaves us with
$$
T(n,q) = n +\frac{2}{n(q+1)} \sum_{j=0}^{n-1} \sum_{k=1}^q  T(j,k).
$$
We now show, using induction on $n$, that $T(n,q)\le 2nH_q$.
\begin{eqnarray*}
T(n,q)& =& n +\frac{2}{n(q+1)} \sum_{k=1}^q \sum_{j=0}^{n-1} T(j,k)
      \quad\le\quad n+\frac{2}{n(q+1)} \sum_{k=1}^q \sum_{j=0}^{n-1} 2jH_k \\
      &\le&  n+\frac{2n}{(q+1)} \sum_{k=1}^q H_k
      \quad=\quad  n+\frac{2n}{(q+1)} \left((q+1)H_q-q \right) 
      \quad\le\quad  2nH_q.
\end{eqnarray*}
This follows, since $\sum_{i=0}^k H_i = (k+1)H_k-k$
and $1/2 \le q/(q+1) <1$.
\end{proof}




\subsection{Search Query Algorithm (\qsearch)\label{sec:search}}

Our input is a value $a$, and the goal is to compute an integer $k$,
  if one exists, such that $a=\A[k]$, with \A\ rearranged so that
  $a$ is in its correct position if \A\ were sorted.
(In other words, we will have $\B[k]=1$ with the invariant preserved.)
\begin{enumerate}
  \item Binary Search Step:
    Starting with $\A[0]$ and $\A[n-1]$ as the left and right
    endpoints ($low=0, high=n-1$), 
    perform the steps of binary search in \A\ for $a$, ignoring
    the fact that \A\ is not sorted.
    Continue until binary search terminates normally, either by finding $a$
    at position $i$ or finding a position $i$ such that $\A[i]<a<\A[i+1]$.
    By the invariant, we have either found $a$, or we have located an
    unsorted subinterval containing $a$ if it is present.
    
  \item \qselect\ Step:
    Proceed as in the \qselect\ algorithm; scan left and right of $i$
    to identify $\ell$ and $r$, and perform a recursive quickselect
    using our chosen \pivot() function.
    We choose which side of the pivot to recurse on based on the value of
     $a$ relative to the pivot.
\end{enumerate}

We consider a \emph{repeated} search query one that takes the binary search 
  step to the same position in the array; 
this can happen even if the query is not for a repeated value.
For repeated searches, the total time is $O(\log n)$ for the binary search,
  and we never perform the \qselect\ step.

The decision not to use the value of $a$ to partition is deliberate and necessary for search queries to be treated, from
  an algorithm analysis perspective, as selection queries with
  $O(\log n)$ additional work.
For $q$ search queries, that additional work is $O(q\log n)$ for the binary
  search steps, which is dominated by the $O(n\log q)$ time done by
  \qselect\ (either in the second step or as a selection query).

%
%
%

\section{Online Cache Oblivious Search and Selection\label{sec:io}}

In this section, we describe an extension to the algorithms for \select\ and \search\ described in Sections~\ref{sec:select}---\ref{sec:search} for the cache-oblivious external memory model~\cite{vitter:io-survey}. To achieve this, we introduce a simple data structure called a \emph{pivot tree}, which stores an up-to-date list of pivots that have been chosen so far. This tree is used to maintain the \search\ query in reasonable time---\select\ does not require the pivot tree at all. The overall algorithm remains relatively simple and is similar to the one described for the internal memory model earlier. We begin with some brief notation for the external memory model.

\paragraph{External Memory Model.} In the external memory model, the computer is abstracted to consist of two memory levels: the internal memory of size~$M$, and the (unbounded) disk memory, which operates by reading and writing data in blocks of size~$B$. We refer to the number of items of the input by~$N$. For convenience, we define $n=N/B$ and $m=M/B$ as the number of blocks of input and memory, respectively. We make the reasonable assumption that $1 \leq B \leq M/2$. In this model, we assume that each I/O read or write is charged one unit of time, and that an internal memory operation is charged no units of time. The \emph{cache-oblivious model} measures performance the same way, but the algorithm has no knowledge of~$M$ or~$B$. To achieve the optimal sorting bound of $SortIO(N) = \Theta(n\log_m n)$ in this setting, it is necessary to make the \emph{tall cache} assumption~\cite{brodal:limits-cache-ob}: $M = \Omega(B^{1+\epsilon})$, for some constant~$\epsilon > 0$. In practice, $M = \Omega(B^2)$, and we will make this assumption for the remainder of the paper.

\subsection{Pivot Tree}
Let the pivot tree~$P$ be a dynamic cache-oblivious $B$-tree~\cite{bender:cache-ob-b-tree,bender:cache-ob-dyn-dict,brodal:cache-ob-search-tree} that maintains one entry for each element of~\A\ that has been chosen as a pivot. This is essentially achieved by maintaining a simple binary search tree laid out in van Emde Boas~\cite{boas:design} format, augmented with ordered-file maintenance on its leaves. This tree supports insertions and deletions in 
$O((\log N)/B + \log_B N)$ I/Os and searches in $O(\log_B N)$ I/Os. Each pivot selected by our algorithm is inserted into~$P$; we will show that for every $t(n)$ processing time, we insert~$p = O(\sqrt{t(n)})$ pivots, with an overall cost of $o(t(n))$ additional time.
We will use this tree primarily to support the \search\ query.
Our algorithms are similar to the ones we developed for the internal memory case. We describe differences below:
\begin{itemize}
	\item During the preprocessing phase, insert~$\A[0]$ and~$\A[n-1]$ into~$P$.
	\item Add the following condition to the invariant from Section~\ref{subsec:invariant}: For each array position~$i$, if $\B[i] = \one$, $P$ contains the value~$\A[i]$.
\end{itemize}	
The \select\ query is nearly identical, except that all pivots discovered in a query are also inserted into~$P$. This maintains the invariant for \select. The \search\ query performs its ``binary search'' phase on~$P$ instead of directly on the array~\A. Once the unsorted subinterval is found that contains the query item~$a$, we resort to the \select\ algorithm on that subinterval as before, again inserting pivots into~$P$ as they are found.

\subsection{Pivot Selection}
We partition the set of items into disjoint subsets such that the invariants are maintained for~\A\ and~\B. In internal memory, we choose one new pivot at each recursive step, and take $O(\log N)$ total recursions, each of which takes linear time. We achieve good bounds, since our internal memory algorithm chooses the same pivots as the ones required for quicksort, although not in the same order. For the I/O case, we need to limit the number of recursive steps our algorithms take to achieve $SortIO(N) = \Theta(n\log_m n)$ I/Os in total (after $q=n$ queries), which is the sorting lower bound in external memory. To this end, we will choose $p-1$ partitioning elements so that the partitions are of roughly equal size. When that is the case, the bucket size decreases from one level of recursion to the next by a relative factor of $\Theta(p)$, and thus there are $O(\log_p n)$ levels of recursion.

Now we reason about the size of~$p$. As the items stream through internal memory, they are partitioned into~$p$ buckets in an online manner. Specifically, when a buffer of size~$B$ fills for one of the buckets, its block can be output to disk, and another buffer (block) takes its place. Hence, the maximum number~$p$ of partitioning elements is $\Theta(m)$ (i.e., the maximum number of buffers that will fit in memory). Here, the number of levels of recursion is $\Theta(\log_m n)$. In the final level of recursion, there is no point in having partitions with fewer than $\Theta(m)$ blocks (since we can sort them for free), so we can limit~$p$ to be $O(n/m)$. These two constraints suggest that the desired number~$p$ of partitioning elements is~$\Theta(\min\{m,n/m\}) = \Theta(\sqrt{n})$.

Probabilistic methods for choosing partition elements based on random sampling~\cite{feller:an} are simple and allow us to choose $p=\Theta(\min\{m,n/m\}) = \Theta(\sqrt{n})$ partitioning elements cache-obliviously in $o(n)$ I/Os.
We briefly sketch the idea: Let $d = O(\log p)$. We take a random sample of $dp$ items, sort the sampled items, and then choose every $d$th item as a partitioning element. Each resulting partition has the desired size of $O(N/p)$ items. The number of I/Os needed to choose these partitioning elements is $O(dp + SortIO(dp))$. Since $p = O(\sqrt{n})$ does not depend on~$M$ or~$B$, the I/O bound is $O(\sqrt{n}\log^2 n) = o(n)$ for choosing pivots in this way, and is therefore negligible.

\begin{theorem}
Given an array~\A\ of $N$ unsorted items, there exists an online cache-oblivious data structure that can support $q$ \select\ and \search\ queries in an online fashion, such that \select\ takes $O(q + n\log_m q)$ time over $q$ unique queries, and \search\ takes $O(n\log_m q + q\log_B N)$ time over $q$ queries, when $q \leq N$.
\end{theorem}
\begin{proof}
To bound \select\ correctly, we focus on a single step of the distribution sort, where $p=O(\sqrt{n})$ partitioning elements are chosen using $O(\sqrt{n} \log^2 n)$ I/Os. In addition, we must insert each of these pivots into the pivot tree, requiring $O(\sqrt{n}\log_B N)$ I/Os. Placing unsorted blocks in the correct partitions requires $O(n)$ I/Os. Overall, we spend $O(n + \sqrt{n} (\log^2 n + \log_B N)) = O(n)$ I/Os. Since this $O(n)$ cost is part of the $n\log_m q$ term from distribution sort, the remaining costs are all lower order terms. Finally, we must spend at least one I/O per \select, contributing an additive $O(q)$ I/Os. The \search\ bound includes an additional $O(\log_B N)$ I/Os for traversing the pivot tree for each of the $q$ queries.
\end{proof}

\section{Allowing Updates\label{sec:dynamic}}

In this section, we extend our results for both internal memory and cache-oblivious external memory to the dynamic case. Recall that we are originally given the unsorted list~\A. Let $\A'$ denote the current list. We want to support the following two additional operations:
\begin{itemize}
	\item \ins(a) inserts $a$ into $\A'$;
	\item \del(i) deletes the $i$th (sorted) entry from~$\A'$.
\end{itemize}

Our solution uses the \emph{dynamic bit dictionary} BitIndel data structure from~\cite{gupta:dynamicdictionary}. We define the following set of supported operations on a dynamic bitvector~$B'$ of current length~$n'$ and original length~$n$. The bitvector~$B'$ is stored implicitly, and is primarily used for the sake of discussion. In the rest of this section, we refer to $b$ as either \zero\ or \one, depending on context.
\begin{itemize}
	\item $\rankb(i)$ tells the number of $b$ bits up to the $i$th position in $B'$;
	\item $\selectb(i)$ gives the position in $B'$ of the $i$th $b$ bit;
	\item $\ins_b(i)$ inserts the bit~$b$ in the $i$th position;
	\item $\delete(i)$ deletes the bit located in the $i$th position;
	\item $\flip(i)$ flips the bit in the $i$th position.
\end{itemize}
Note that one can determine the $i$th bit of $B'$ by computing $\rankone(i) - \rankone(i-1)$. (For convenience, we assume that $\rankb(-1) = 0$.)

\begin{lemma}[\cite{gupta:dynamicdictionary}]
Given a bitvector~$B'$ with length~$n'$ and original length~$n$, BitIndel is a cache-oblivious data structure that takes $O(n')$ bits and supports
$\rank$ and $\select$ in $O(\log_n n')$ time, and \ins\ and \delete\ in
$O(n^\epsilon \log_n n')$ amortized time. When $n' = o(n)$, the
bounds become $O(1)$ and $O(n^\epsilon)$ respectively, for any constant~$0 < \epsilon < 1$.
\end{lemma}

We will use a similar technique to dynamize both internal and external memory scenarios. We will maintain the list~$\A$ (although we will swap values as usual) and retain the bitvector~$\B$ that marks which entries of~$\A$ are pivots. We define an \emph{insert bitvector}~$I$ such that $I[i] = \one$ if and only if $\A'[i]$ is newly inserted. Similarly, we define a \emph{delete bitvector}~$D$ such that if $D[i] = \one$, the $i$th item in $\A$ has been deleted. If a newly inserted item is deleted, it is removed from~$I$ directly.
Both $I$ and $D$ are implemented as instances of the BitIndel data structure.

We will also use a data structure~$T$ that maintains the value and position in~$A'$ of all newly-inserted items. We manage~$T$ similarly to what was done in~\cite{gupta:dynamicdictionary} for the BitIndel data structure. We define~$T$ to be a Weight-Balanced B-tree (WBB) with fanout between $[n^\epsilon, 2n^\epsilon]$ for a fixed $\epsilon > 0$. The leaves of this tree maintain contiguous chunks of an array of all newly-inserted items, in value (and thus position) order. Each entry in the array maintains both its value and its correct position in the array~$A'$. The tree~$T$ is balanced by \emph{array index among newly-inserted items}, or in other words, the rank of each item in the array. 

Each internal node of the tree maintains a prefix-sum data structure that can quickly direct a query to the appropriate branch. This can be done in $O(1)$ time per node according to~\cite{Hon}. At the bottom level, each ``leaf'' node is an array of contiguous positions; we access them directly in $O(1)$ time as usual. When we update~$T$, it takes at most $O((1/\epsilon)n^\epsilon)$ time, as per the discussion in~\cite{gupta:dynamicdictionary}.

Our preprocessing steps are the same as the static cases. 
The bitvectors~$I$ and~$D$ are storing the implicit bitvector of $n$ \zero s.
The WBB tree~$T$ is initially empty. 
Next, 
we describe how to support $A'.\ins(a)$, $\A'.\del(i)$, $A'.\select(i)$, and $A'.\search(a)$ for the internal memory case.

\subsection{Internal Memory Dynamic Online Selection and Search}
\label{subsec:internal-dynamic}



\paragraph{$A'.\ins(a)$.} First, we search for the appropriate unsorted interval~$[\ell,r]$ containing~$a$ using a binary search on the original (unsorted) array~$A$.
Now perform \qsearch(a) on interval~$[\ell,r]$ (choosing which subinterval to expand based on the insertion key~$a$) until $a$'s exact position~$t$ in $\A$ is determined. The original array~$A$ must have chosen as pivots the elements immediately to its left and right (positions~$t-1$ and $t$ in array~$\A$); hence, one never needs to consider newly-inserted pivots when choosing subintervals.
Insert~$a$ in sorted order in~$T$ among all other newly-inserted items. Calculate $t'=I.\selectzero(t)$, and set $a$'s position to $t''=t'-D.\rankone(t')$. Finally, we update our bitvectors by performing $I.\ins_\one(t'')$ and $D.\ins_\zero(t'')$.

\paragraph{$\A'.\delete(i)$.} First, we compute~$i'=D.\selectzero(i)$. If $i'$ is newly-inserted (i.e., $I[i'] = \one$), then we scan through $T$ and remove the item with position~$i' - D.\rankone(i')$, and reduce the position of all newly-inserted items after that position by 1. Then we perform $I.\delete(i')$ and $D.\delete(i')$. If instead $i'$ is an older entry, we update $T$ to reduce the position of all newly-inserted items after position~$i'-D.\rankone(i')$ and perform $D.\flip(i')$.\footnote{If a user wants to delete an item with value~$a$, one could simply search for it first to discover its rank, and then delete it using this function. Curiously enough, this means we can delete items that haven't yet been placed correctly!}

\paragraph{$A'.\select(i)$.} If $I[i] = \one$, go to the $t=I.\rankone(i)$ entry of $T$ and report the value there. Otherwise, compute $t=I.\rankzero(i)-D.\rankone(i)$. Now, perform the \qselect(t).

\paragraph{$A'.\search(a)$.} First, we search for the appropriate unsorted interval~$[\ell,r]$ containing~$a$ using a binary search on the original (unsorted) array~$A$. Then, perform \qsearch(a) on interval~$[\ell,r]$ until $a$'s exact position~$t$ is found. If $a$ appears in array~$A$ (which we discover through \qsearch), we need to now check whether it has been deleted. We compute $t' = I.\selectzero(t)$ and $t'' = t' - D.\rankone(t')$. If $D[t'] = \zero$, return $t''$. Otherwise, it is possible that the item has been newly-inserted. Compute $p=I.\rankone(t')$, which is the number of newly-inserted items. If $T[p-1] = a$ or $T[p]=a$, return its position (which is stored in $T$). Otherwise, return failure.

\paragraph{Time Complexity.} Since every operation except \delete\ generates pivots by calling \qselect\ or \qsearch, they contribute to the $O(n\log q)$ overall time required over $q$ queries. Each access to $T$ is done with an explicit rank, and so these queries are finished in $O(1/\epsilon)$ time each. The functions \delete\ and \ins\ contribute an additive $O((1/\epsilon)n^\epsilon)$ time per operation to perform BitIndel insertions and deletions, as well as manipulation (and possible rewriting) of~$T$. Thus, we arrive at the following theorem.

\begin{theorem}
Given a dynamic array~$\A'$ of $n$ original items and $n' = o(n)$ current items, there exists a dynamic online internal memory data structure that can support $q \leq n$ \select\ and \search\ queries in an online fashion, such that \select\ takes $O(n\log q)$ time over $q$ unique queries and \search\ takes $O(n\log q + q\log n)$ time over $q$ queries. Insertions and deletions to~\A\ take $O((1/\epsilon)n^\epsilon)$ time each.
\end{theorem}

\subsection{Cache-Oblivious Dynamic Online Selection and Search}
\label{subsec:external-dynamic}

For this, we augment our data structure in Section~\ref{sec:io} with the WBB tree~$T$ on newly inserted items. We can implement~$T$ in a cache-oblivious way using a van Emde Boas layout, together with ordered-file maintenance on leaves. The prefix-sum data structure in each node of~$T$ is cache-oblivious~\cite{gupta:dynamicdictionary}. Since all of the data structures above are already cache-oblivious, it is sufficient to show that we do not incur any unnecessary I/Os during our \search\ and \select\ operations. Moreover, when we rewrite entire portions of~$T$, these portions are contiguous blocks of memory, and require no more than $O((1/B\epsilon)n^\epsilon)$ I/Os.

\begin{theorem}
Given a dynamic array~$\A'$ of $N$ original items and $N' = o(N)$ current items, there exists a dynamic online cache-oblivious data structure that can support $q \leq N$ \select\ and \search\ queries in an online fashion, such that \select\ takes $O(q + n\log_m q)$ I/Os and \search\ takes $O(n\log_m q + q\log_B N)$ I/Os over $q$ unique queries. Insertions and deletions to~\A\ take $O((1/B\epsilon)n^\epsilon)$ I/Os.
\end{theorem}
\begin{proof}
Since we are choosing the same pivots as the external distribution sort from  Section~\ref{sec:io}, we are doing our initial \search\ on the pivot tree~$P$, and all data structures and their access is contiguous, the theorem is proved.
\end{proof}

\section{Online Data Compression\label{sec:compress}}

In this section, we briefly outline results to expand the static online selection problem to compress and sort simultaneously. To do this, we group chunks of $\log^2 n$ items together into a BSGAP~\cite{bsgapjournal} data structure. In this structure, one can encode the underlying data in terms of the ``gaps'' between contiguous items, and still binary search the items contained therein.



We describe our data structure in a top-down way. At the top level, we will build a pivot tree, implemented as an augmented Andersson-Thorup predecessor data structure~\cite{andersson00worstcase} containing no more than 
$O(n/\log^2 n)$ pivots. Each pivot will maintain its rank among other items, in addition to its value. (This rank information can be maintained within the exponential tree.) Overall, the pivot tree will require $O(n/\log n) = o(n)$ bits in total. This pivot tree is initially empty. 

As the $\pivot()$ function is called and we promote items to pivots, we build BSGAP data structures if we encounter an unsorted interval of size~$O(\log^2 n)$. The root of this BSGAP structure is ``gap encoded'' according to its left or right endpoint of the interval, as appropriate. Furthermore, we maintain a pointer from the pivot to this BSGAP structure using $O(\log n)$ bits of space.

If the interval is larger, we must partition the elements according to the pivot element and recurse. During this process, we encode each item according to its best ancestor. This best ancestor must either be the left or right endpoint (pivot) of its unsorted interval; this information can be stored with a single bit.

We defer the details until the full version of the paper; these results also have a dynamic analogue, although not as compressed. We present some results for the average case analysis in the static case.

\begin{theorem}
Given an array~\A\ of $n$ unsorted items drawn from a universe of size~$u$, there exists an online data structure that can support $q$ \select\ and \search\ queries in an online fashion, such that \search\ and \select\ require $O(AT(u,n) + n\log q)$ time over $q$ queries, when $q \leq n$. In addition, the space of the data structure eventually diminishes to $gap + O(n\log\log(u/n))$ bits as a natural consequence of answering queries.
\end{theorem}

\section{Implementation Results\label{sec:implement}}

In this section we present some preliminary experiments.
Our goal is to show that the \qselect\ and \qsearch\ algorithms
  are practical and reasonably efficient.

\paragraph{\qselect.}
In the table below,
  the first row gives the times for sorting using \qselect\ by using
  the ranks $k=1,3,5,\ldots$. 
This should be competitive, but slower than, quicksort. 
The second row shows results for $\sqrt{n}$ randomly chosen selections 
  using \qselect; this \textit{should} take less time than sorting.
  ($O(n\log q)$ time with $q=\sqrt{n}$ versus $O(n\log n)$ time). 
Finally, we timed Quicksort. 
Our results are as expected.
\begin{center}
\begin{tabular}{l|rrrr}
\hline
\noalign{\smallskip}
Method & $n=$ &  $10^5$ & $10^6$ & $10^7$\\
\noalign{\smallskip}
\hline
\noalign{\smallskip}
$k=1,3,5,\ldots$ to sort &  & 0.021 & 0.240 & 2.78\\[3pt]
$\sqrt{n}$ selections & & 0.010 & 0.120 & 1.32\\[3pt]
Quicksort &  & 0.016 & 0.188 & 2.16\\
\noalign{\smallskip}
\hline
\end{tabular}
\end{center}
For all methods, we used lists of $n$ randomly chosen integers 
  from the range $0,\ldots,n-1$.
Times are reported in average number of CPU seconds. 
We performed the experiments on an Intel Pentium D, 
  3.0GHz with 1MB of cache, running Linux and using the GNU g++ compiler 
  with level 2 optimization.
We did not bother to pack our bit vector.


\paragraph{\qsearch.}
In the table below, we compare the times for two algorithms
  for online searching:
\begin{itemize}
  \item \textit{Old Method}: Perform a quicksort, then use
     interpolation search to find items.
  \item \textit{New Method}: Use \qsearch,
     with interpolation search in place of binary search.  
\end{itemize}

\begin{center}
\begin{tabular}{cl|crrr}
\hline
\noalign{\smallskip}
Search Count&Method& $n=$ &  $10^5$ & $10^6$ & $10^7$\\
\noalign{\smallskip}
\hline
\noalign{\smallskip}
$q=\sqrt{n}$ 
&Old & &  0.0164 & 0.198 & 2.22\\[3pt]
&New & &  0.0100 & 0.114 & 1.27\\
\noalign{\smallskip}
\hline
\noalign{\smallskip}
$q=n/10$ 
&Old &  &  0.0196 & 0.240 & 2.87\\[3pt]
&New &  &  0.0192 & 0.263 & 3.12\\
\noalign{\smallskip}
\hline
\end{tabular}
\end{center}

Search keys were chosen uniformly at random.
Using interpolation search in place of binary search should favor
  the Old Method, as the data are uniformly distributed.
The results we obtained using binary search were similar
  to what we report here. We obtained times for $q=\sqrt{n}$ searches, which should favor our
  new algorithm (like above), and for $q=n/10$ searches, which should favor the old algorithm, since we are almost completely sorting the list.



\bibliographystyle{plain}


\end{document}